\documentclass[%
 reprint,
 amsmath,amssymb,
 aps,
pra,
showkeys,
]{revtex4-2}

\usepackage{graphicx}
\usepackage{dcolumn}
\usepackage{bm}

\usepackage{physics}
\usepackage{amsthm}
\usepackage{mathtools}

\theoremstyle{theorem}
\newtheorem{theorem}{Theorem}
\newtheorem{lemma}{Lemma}
\newtheorem{corollary}{Corollary}[lemma]

\theoremstyle{remark}
\newtheorem*{remark}{Remark}

\theoremstyle{plain}

\theoremstyle{definition}

\newcommand{\hil}{\mathcal H}
\newcommand{\sw}{\mathrm{SWAP}}

\begin{document}

\preprint{APS/123-QED}

\title{Construction of genuinely entangled multipartite subspaces from bipartite ones by reducing the total number of separated parties
}

\author{K. V. Antipin}
 \email{kv.antipin@physics.msu.ru}
 \affiliation{Faculty of Physics, M. V. Lomonosov Moscow State University,\\ Leninskie gory, Moscow 119991, Russia}

\date{\today}

\begin{abstract}
Construction of genuinely entangled multipartite subspaces with certain characteristics has become a relevant task in various branches of quantum information. Here we show that such subspaces can be obtained from an arbitrary collection of bipartite entangled subspaces under joining of their adjacent subsystems. In addition, it is shown that direct sums of such constructions under certain conditions are genuinely entangled. These facts are then used in detecting entanglement of tensor products of mixed states and constructing subspaces that are distillable across every bipartite cut, where for the former application we include an example with the analysis of genuine entanglement of a tripartite state obtained from two Werner states.  
\end{abstract}

\maketitle


\section{Introduction}

Entangled subspaces have become an object of intensive research in recent years due to their  potential utility in the tasks of  quantum information processing. Ref.~\cite{Parth04}, the work by K. R. Parthasarathy, where completely entangled subspaces~(CESs) were described, can be thought of as a starting point for developing this direction. CESs are subspaces that are free of fully product vectors. This concept  was later generalized to genuinely entangled subspaces~(GESs)~\cite{DemAugWit18,CMW08} -- those entirely composed of states in which entanglement is present in every bipartite cut of a compound system.

Genuine multipartite entanglement~(GME), being the strongest form of entanglement, has found many applications in quantum protocols~\cite{YeCh06,MP08,MEO18}. In this connection genuinely entangled subspaces  are useful since they can serve as a source of GME states. As an example, it is known that any state entirely supported on GME is genuinely entangled. Another example is connected with detection of genuine entanglement: a state having significant overlap with a GES is genuinely entangled~\cite{DemAugQut19,KVAnt21}, and certain entanglement measures can be estimated for such a state~\cite{KVAnt21}. There are also some indications that GESs can be used in quantum cryptography~\cite{SheSrik19} and quantum error correction~\cite{HuGra20}.

There are several approaches to construction of GESs~\cite{DemAugWit18, AgHalBa19, DemAug20, KVAnt21,Dem21}, including those of maximal possible dimensions. While the problem of constructing maximal GESs for any number of parties and any local dimensions  seems to be solved recently in Ref.~\cite{Dem21}, it is of significant interest to build entangled subspaces with certain useful for quantum protocols characteristics such as given values of entanglement measures, distillability property, robustness of entanglement under external noise, etc. It is the task we concentrate on in the present paper, following the path of compositional construction started in Ref.~\cite{KVAnt21}. We investigate a special operation when bipartite completely entangled subspaces are combined together with the use of tensor products with subsequent joining  the adjacent subsystems~(parties). We show that such an operation can generate GESs and that its compositional character together with the freedom of choice of the input subspaces opens the possibility to control the parameters of the output GESs. Such construction can be relevant for  quantum networks~\cite{Sim17,BFD19,KSYG21}. In particular, when two states are combined, this operation corresponds to the star configuration~\cite{TenNet2021}. Combination of two subspaces in turn can be associated with a superposition of several quantum networks.

The paper is structured as follows. In Section~\ref{sec::prel} we give necessary definitions and provide some mathematical background. In Section~\ref{sec::main} the main lemmas concerning the properties of tensor products of entangled subspaces are stated and proved. In Section~\ref{sec::app} it is shown how the established properties can be applied in several tasks such as constructing GESs with certain useful properties, detecting  entanglement of tensor products of mixed states. In Section~\ref{sec::disc} we conclude and propose possible directions of further research.

\section{Preliminaries}\label{sec::prel}

Throughout this paper we consider finite dimensional Hilbert spaces and their tensor products. We begin with more precise definitions of entangled states and subspaces.

A pure $n$-partite state is \emph{entangled} if it cannot be written as a tensor product of states for every subsystem, i.~e.,
\begin{equation}
    \ket{\psi}\ne\ket{\phi}_1\otimes\ldots\otimes\ket{\phi}_n.
\end{equation}

A \emph{bipartite cut}~(\emph{bipartition}) $A|\bar A$ of an $n$-partite state is defined by specifying  a subset $A$ of the set of $n$ parties as well as its complement $\bar A$ in this set.

A pure $n$-partite state $\ket{\psi}$ is called \emph{biseparable} if it can be written as a tensor product
\begin{equation}
    \ket{\psi} = \ket{\phi}_A\otimes\ket{\chi}_{\bar{A}}
\end{equation}
with respect to some bipartite cut $A|\bar A$. On the contrary, a multipartite pure state is called \emph{genuinely entangled} if it is not biseparable with respect to any bipartite cut.

Similarly, a mixed  multipartite state is called \emph{biseparable} if it can be decomposed into a convex sum of biseparable pure states, not necessarily  with respect to the same bipartite cut. In the opposite case it is called \emph{genuinely  entangled}. 

A subspace of a multipartite Hilbert space is called \emph{completely entangled}~(CES) if it consists only  of entangled states. A \emph{genuinely  entangled subspace}~(GES) is a subspace composed entirely of genuinely entangled  states.

Next we recall some measures of entanglement.

The geometric measure of entanglement of a bipartite pure state $\ket{\psi}$ is defined by
\begin{equation}\label{geom}
    G(\psi) \coloneqq 1 - \max_i\{\lambda_i\},
\end{equation}
where $\lambda_i$ is the $i$-th Schmidt coefficient squared as in the Schmidt decomposition $\ket{\psi} = \sum_i\,\sqrt{\lambda_i}\ket{i}\otimes\ket{i}$. This measure is generalized~\cite{Guhnetall20} to detect genuine multipartite entanglement as
\begin{equation}\label{geomGME}
    G_{GME}(\psi) \coloneqq \min_{A|\bar A} G_{A|\bar A}(\psi),
\end{equation}
where the minimization runs over all possile bipartite cuts $A|\bar A$ and $G_{A|\bar A}(\psi)$ -- the geometric measure~(\ref{geom}) with respect to bipartite cut $A|\bar A$.

For mixed multipartite states the geometric measure of genuine entanglement is defined via the convex roof construction
\begin{equation}\label{cfGME}
    G_{GME}(\rho) \coloneqq \min_{\{(p_j,\,\psi_j)\}}\,\sum_j\,p_j\,G_{GME}(\psi_j),
\end{equation}
where the minimum is taken  over all  ensemble decompositions $\rho = \sum_j p_j\dyad{\psi_j}$.

To quantify entanglement of a subspace $\mathcal S$, we will use the entanglement measure $EM$ of its least entangled vector:
\begin{equation}
   EM(\mathcal S) \coloneqq  \min_{\ket{\psi}\in\mathcal S}EM(\psi).
\end{equation}
In place of $EM$ here can be used the geometric measure $G_{A|\bar A}$ across a specific bipartite cut, as well as the genuine entanglement measure $G_{GME}$ of Eq.~(\ref{geomGME}).

We proceed to quantum channels and their connections with entangled subspaces.

Let $\mathcal L(\hil)$ denote the set of all linear operators on  $\hil$. 
A \emph{quantum channel} $\Phi_{A\rightarrow B}$ is a linear, completely positive and trace-preserving map between $\mathcal L(\hil_A)$ and $\mathcal L(\hil_B)$ \cite{Wilde13}, for two finite dimensional Hilbert spaces $\hil_A$ and $\hil_B$.

A crucial property used in the present work is the  correspondence between quantum channels and linear subspaces of composite Hilbert spaces~\cite{AubSz17}. Consider an isometry $V\,\colon\, \hil_A\rightarrow \hil_B\otimes \hil_C$ whose range is $W$, some subspace of $\hil_B\otimes\hil_C$. The corresponding quantum channel $\mathrm\Phi_{A\rightarrow B}\colon\,\mathcal L(\hil_A)\rightarrow\mathcal L(\hil_B)$ can be introduced by
\begin{equation}\label{IsoRep}
    \mathrm\Phi_{A\rightarrow B}(\rho) = \mathrm{Tr}_{\hil_C} (V\rho V^{\dagger}).
\end{equation}
If we trace out subsystem $B$ instead, a \emph{complementary}~\cite{ShorDev} to $\Phi$ quantum channel $\Phi^C_{A\rightarrow C}$ is obtained:
\begin{equation}\label{IsoRepCom}
    \mathrm\Phi^C_{A\rightarrow C}(\rho) = \mathrm{Tr}_{\hil_B} (V\rho V^{\dagger}).
\end{equation}
The correspondence works in the opposite direction as well: by Stinespring's dilation theorem~\cite{Stine55}, for any channel $\mathrm\Phi_{A\rightarrow B}$ there exists some subspace $W\subset \hil_{B}\otimes \hil_{C}$ such that $\mathrm\Phi_{A\rightarrow B}$ is determined by Eq.~(\ref{IsoRep}).

Eqs.~(\ref{IsoRep}) and (\ref{IsoRepCom}) are represented diagrammatically on Fig.~\ref{fig:chancom}. In this paper we use tensor diagram notation and the corresponding tools for diagrammatic reasoning from Ref.~\cite{CoeKis17}, which include the discarding symbol depicting tracing out a particular subsystem and various line deformations denoting linear algebra operations.   Refs.~\cite{BiamonteEtAll15, Biamonte19} are also good sources on  application of tensor diagrams in quantum information theory.

\begin{figure}[t]
\includegraphics[scale=0.30]{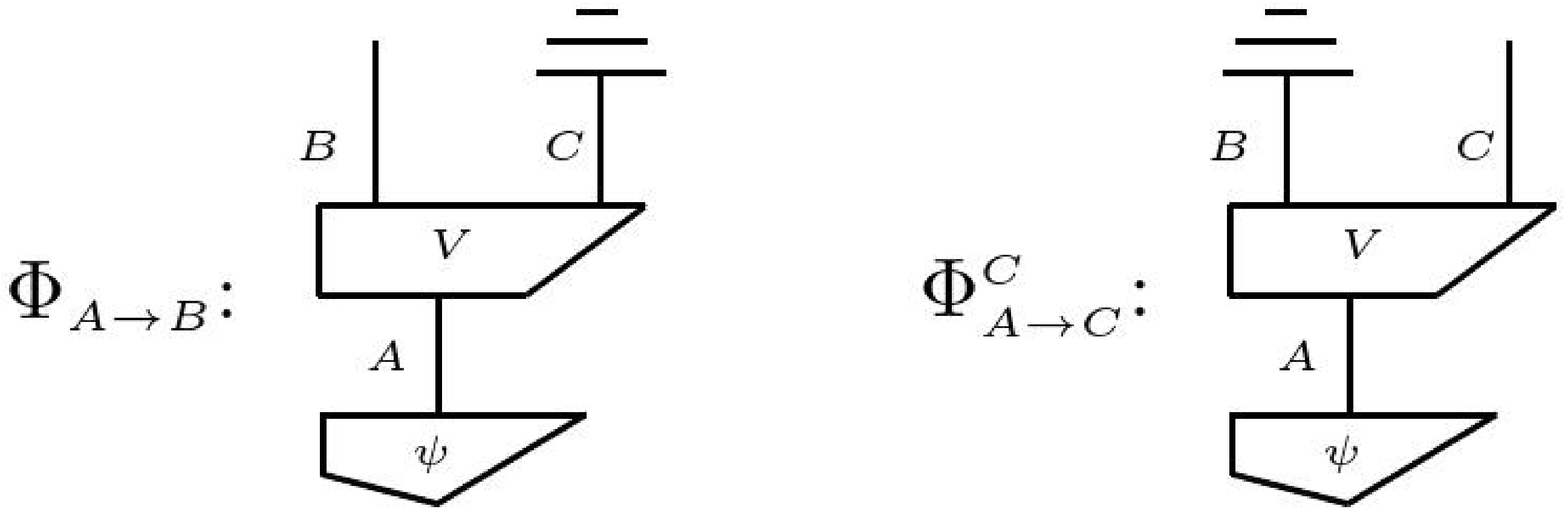}
\caption{\label{fig:chancom} A representation of channel $\mathrm\Phi_{A\rightarrow B}$ and its complementary channel $\Phi^C_{A\rightarrow C}$ both acting on a pure state $\ket{\psi}\in\hil_A$: the isometry $V$ takes the state  to $\hil_B\otimes \hil_C$, then one of the two subsystems is traced out~(which is  denoted by the  discarding symbol).}
\end{figure}

An important characteristic of a quantum channel $\Phi$ is the \emph{maximal output norm}~\cite{AmHolWer2000} defined by
\begin{equation}\label{OutNorm}
    \nu_p(\mathrm\Phi) = \sup_{\rho\in\mathcal D(\hil)}\norm{\mathrm\Phi(\rho)}_p,\,\,p>1,
\end{equation}
where $\norm{\rho}_p = (\mathrm{Tr}(\abs{\rho}^p))^{1/p}$ is the  $p$-norm and $\mathcal D(\hil)$ is the set of density operators on $\hil$. The supremum in Eq.~(\ref{OutNorm}) can be taken over pure input states due to  convexity of the $p$-norm. The quantity $\nu_p(\mathrm\Phi)$ also characterizes the entanglement of the subspace $W$ corresponding to the channel $\Phi$: $W$ is completely entangled iff $\nu_p(\mathrm\Phi) < 1$.

\begin{figure}[t]
\includegraphics[scale=0.35]{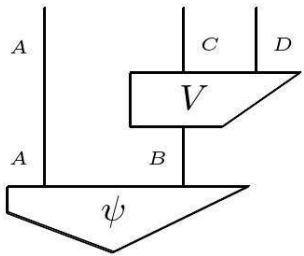}
\caption{\label{fig:iso23} An isometry $V$ acting on  subsystem $B$ of  a pure bipartite  state $\ket{\psi}$ from a completely entangled subspace of a tensor product Hilbert space $\hil_A\otimes\hil_B$. Acting of a properly chosen isometry on each state in the subspace generates a genuinely entangled subspace of a tripartite Hilbert space $\hil_A\otimes\hil_C\otimes\hil_D$. }
\end{figure}

Let us mention another crucial property concerning the maximal output norm. Consider a product channel $I\otimes\Phi$, where $I$ is the identity map~(the ideal channel). Then 
\begin{equation}\label{idealcomb}
    \nu_p(I\otimes\Phi) = \nu_p(\Phi),\quad 1\leqslant p\leqslant\infty.
\end{equation}
It was proved in Ref.~\cite{AmHolWer2000}.

Ref.~\cite{KVAnt21} provides a simple approach to constructing tripartite genuinely entangled subspaces with the use of composition of bipartite completely entangled subspaces and quantum channels of certain types. The approach is presented on Fig.~\ref{fig:iso23}, where an isometry $V$ is acting on one of the two subsystems of each state from a completely entangled subspace of $\hil_A\otimes\hil_B$. It was shown that, when the isometry corresponds to a quantum channel $\Phi$ with $\nu_p(\mathrm\Phi) < 1$ for  $p>1$~(i.~e., the isometry has a CES as its range), a genuinely entangled subspace of $\hil_A\otimes\hil_C\otimes\hil_D$ is generated.

Interestingly enough, there are other types of isometries that can generate GESs via the scheme on Fig.~\ref{fig:iso23}, and they don't necessarily have completely entangled ranges. In the present paper, though, we will use those of the described above type.

There will be a lot of joining of subsystems in the present paper. Let $A$ and $B$ be two systems with Hilbert spaces $\hil_A$ and $\hil_B$, respectively, and $\dim(\hil_A)=d_A$, $\dim(\hil_B)=d_B$.   Let $C$ be a larger system such that $\dim(\hil_C)=d_A d_B$. We say that   $A$ and $B$ are joined into  $C=AB$ if, given fixed computational bases $\{\ket{i}_A\}$ and $\{\ket{j}_B\}$ of $\hil_A$ and $\hil_B$ respectively, there is a mapping between the product basis of $\hil_A\otimes\hil_B$ and  a fixed computational basis $\{\ket{k}_C\}$ of $\hil_C$:
\begin{equation}
    \ket{i}_A\otimes\ket{j}_B\rightarrow\ket{k'}_C,\quad
    k'=i\,d_B  + j,
\end{equation}
i.~e., the bases are joined in the lexicographic order.
The mapping is extended on all other vectors of $\hil_A\otimes\hil_B$ by linearity.

\section{\label{sec::main}Entangled states and subspaces from tensor product}
We begin the section with a simple observation.
\begin{lemma}
Let $\ket{\phi}_{AB_1}$ and $\ket{\chi}_{B_2C}$ be two pure bipartite entangled states on $\hil_A \otimes\hil_{B_1}$ and $\hil_{B_2}\otimes\hil_C$, respectively. Let $\ket{\psi}_{ABC}$ be a tripartite pure state on $\hil_A\otimes\hil_B\otimes\hil_C$ that is obtained from taking the tensor product $\ket{\phi}_{AB_1}\otimes\ket{\chi}_{B_2C}$ with subsequent joining  subsystems $B_1$ and $B_2$ into a larger one, $B=B_1B_2$~(see Fig.~\ref{fig:twostates}). Then $\ket{\psi}_{ABC}$ is genuinely entangled.
\end{lemma}
\begin{proof}
One needs to check that the tripartite state is entangled across all three bipartitions $A|BC$, $B|AC$, $C|AB$, which can be conveniently seen from the diagrammatic representation. For bipartition $B|AC$, as shown  on Fig.~\ref{fig:twostatesbip}, tracing out subsystem $B$~(i. e., subsystems $B_1$ and $B_2$) results in a state equal to $\rho^{\phi}_A \otimes \rho^{\chi}_C$, where 
\begin{equation}\label{parden}
 \rho^{\phi}_A = \mathrm{Tr}_{B_1}\{\dyad{\phi}_{AB_1}\},\quad \rho^{\chi}_C = \mathrm{Tr}_{B_2}\{\dyad{\chi}_{B_2C}\}.
\end{equation}
The bipartite states  $\ket{\phi}_{AB_1}$ and $\ket{\chi}_{B_2C}$ are entangled, and hence the corresponding one party states $\rho^{\phi}_A$ and $\rho^{\chi}_C$ are mixed. As a tensor product of mixed states, the resulting state is also mixed. The other two bipartitions are analyzed similarly.
\end{proof}

\begin{figure}[t]
\includegraphics[scale=0.25]{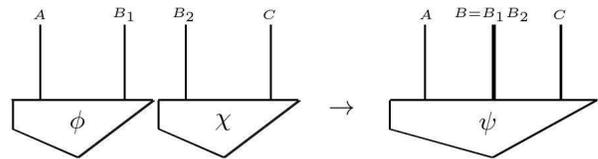}
\caption{\label{fig:twostates} Tensor product of two bipartite entangled pure states generates a tripartite genuinely entangled state after joining subsystems $B_1$ and $B_2$.}
\end{figure}

\begin{figure}[b]
\includegraphics[scale=0.23]{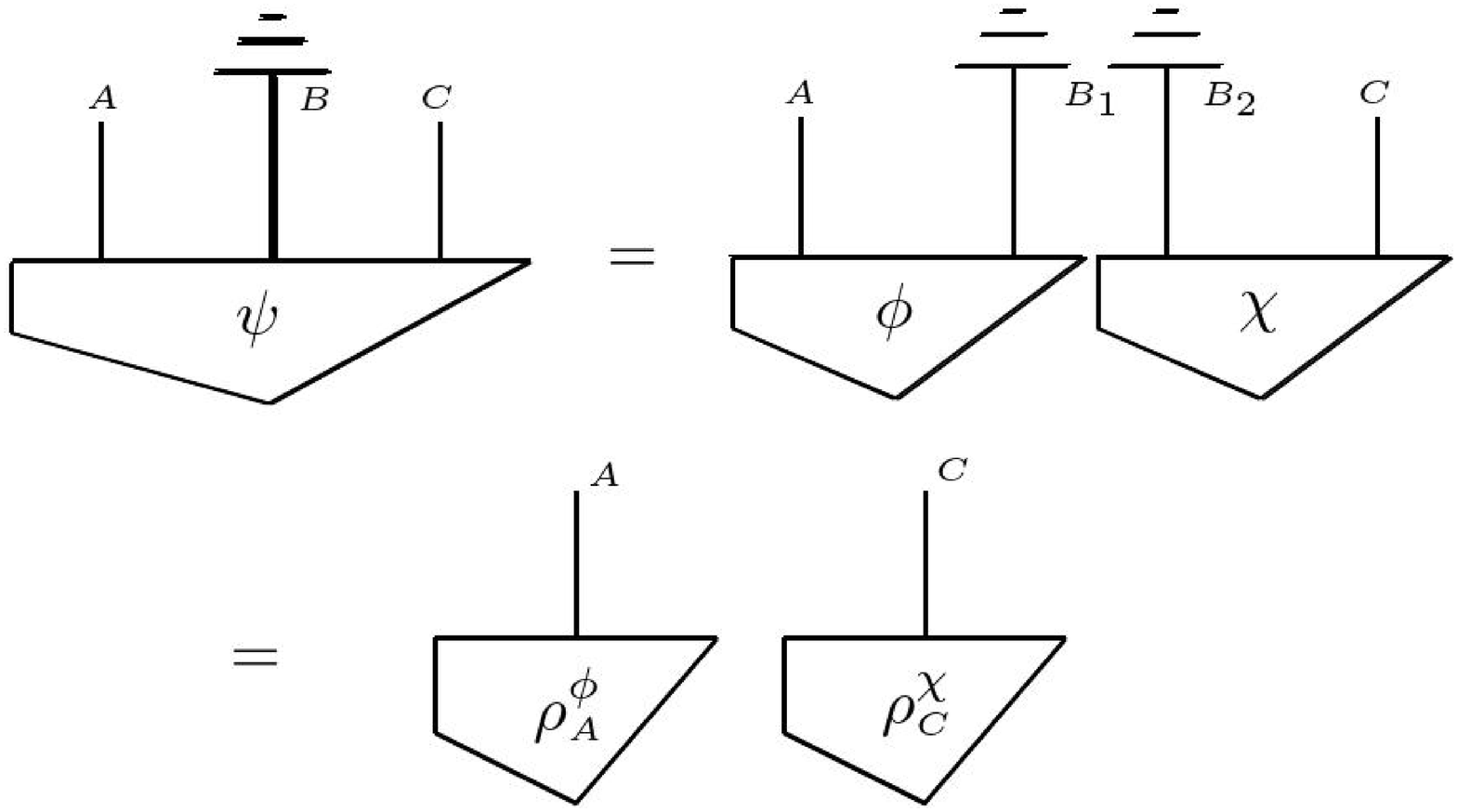}
\caption{\label{fig:twostatesbip} Entanglement in bipartition $B|AC$ of a tripartite pure state $\ket{\psi}_{ABC}$: after tracing out subsystem $B$ the resulting state is a tensor product of two mixed states $\rho^{\phi}_A$ and $\rho^{\chi}_C$.}
\end{figure}

What is more interesting is that two bipartite entangled subspaces can be combined in a similar way to generate a genuinely entangled subspace. 

\begin{lemma}\label{gesprod}
Let $\mathcal S_{AB_1}$ be a completely entangled  subspace of $\hil_A\otimes\hil_{B_1}$, and $\mathcal G_{B_2C}$ -- a completely entangled  subspace of $\hil_{B_2}\otimes\hil_C$.
Then their tensor product $\mathcal S_{AB_1}\otimes\mathcal G_{B_2C}$, after joining subsystems $B_1$ and $B_2$ into $B=B_1B_2$, is a genuinely entangled subspace of $\hil_A\otimes\hil_B\otimes\hil_C$, with the geometric measure of genuine entanglement 
\begin{equation}\label{gmeges}
     G_{GME}(\mathcal S_{AB_1}\otimes\mathcal G_{B_2C}) = \min\left(G(\mathcal S_{AB_1}),\,G(\mathcal G_{B_2C})\right).
\end{equation}
\end{lemma}

\begin{proof}
The argument follows from diagrammatic reasoning involving the correspondence between bipartite subspaces and quantum channels.

Let $\ket{\psi_1}_{AB_1},\,\ldots,\,\ket{\psi_n}_{AB_1}$ be basis vectors in $\mathcal S_{AB_1}$, and $\ket{\chi_1}_{B_2C},\,\ldots,\,\ket{\chi_k}_{B_2C}$ -- basis vectors in $\mathcal G_{B_2C}$.  The elements  $\{\ket{\psi_i}_{AB_1}\otimes\ket{\chi_j}_{B_2C}\}$ then span $\mathcal S_{AB_1}\otimes\mathcal G_{B_2C}$.

Consider also Hilbert spaces $\hil_D$ and $\hil_E$ with $\dim(\hil_D)=\dim(\mathcal S_{AB_1})$, $\dim(\hil_E)=\dim(\mathcal G_{B_2C})$ and basis states $\ket{\mu_1}_{D},\,\ldots,\,\ket{\mu_n}_{D}$ and $\ket{\nu_1}_{E},\,\ldots,\,\ket{\nu_k}_{E}$, respectively.

Let $V_1\colon\,\hil_{D}\rightarrow\hil_{A}\otimes\hil_{B_1}$ be an isometry that maps the states $\{\ket{\mu_i}_{D}\}$ to the states $\{\ket{\psi_i}_{AB_1}\}$, and $V_2\colon\,\hil_{E}\rightarrow\hil_{B_2}\otimes\hil_{C}$ -- an isometry mapping $\{\ket{\nu_j}_{E}\}$ to $\{\ket{\chi_j}_{B_2C}\}$.
The ranges of $V_1$ and $V_2$ are then the completely entangled subspaces $\mathcal S_{AB_1}$ and $\mathcal G_{B_2C}$, respectively.

A particular element $\ket{\psi_i}_{AB_1}\otimes\ket{\chi_j}_{B_2C}$ of $\mathcal S_{AB_1}\otimes\mathcal G_{B_2C}$ can be written as
\begin{multline}
    \ket{\psi_i}_{AB_1}\otimes\ket{\chi_j}_{B_2C} \\
    = \left(V_1\otimes V_2\right)\left(\ket{\mu_i}_{ D}\otimes\ket{\nu_j}_{E}\right),
\end{multline}
and hence the whole subspace $\mathcal S_{AB_1}\otimes\mathcal G_{B_2C}$ can be presented  as the result of action of the isometry $V_1\otimes V_2$ on each state from the tensor product Hilbert space $\hil_D\otimes\hil_E$ spanned by $\{\ket{\mu_i}_{D}\otimes\ket{\nu_j}_{E}\}$~(see Fig.~\ref{fig:gesgenod}).

\begin{figure}[t]
\includegraphics[scale=0.35]{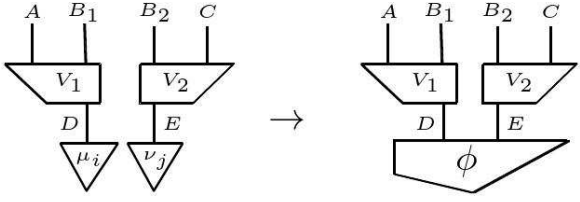}
\caption{\label{fig:gesgenod} To the left: action of the isometry $V_1\otimes V_2$ on a particular basis state $\ket{\mu_i}_{D}\otimes\ket{\nu_j}_{E}$ of $\hil_D\otimes\hil_E$. To the right: action of the isometry $V_1\otimes V_2$ on a general  state $\phi$ from $\hil_D\otimes\hil_E$, which is equal to a linear combination  of basis states $\{\ket{\mu_i}_{D}\otimes\ket{\nu_j}_{E}\}$. Action of   $V_1\otimes V_2$ on each state in $\hil_D\otimes\hil_E$ generates $\mathcal S_{AB_1}\otimes\mathcal G_{B_2C}$.  }
\end{figure}

We can use this diagrammatic representation of a general state from $\mathcal S_{AB_1}\otimes\mathcal G_{B_2C}$ for the analysis of entanglement.
Consider now bipartition $A|BC$. Tracing out subsystems $B=B_1B_2$ and $C$ of the state has the same effect as tracing out subsystem $E$ of the corresponding  state  $\ket{\phi}_{DE}$ from $\hil_D\otimes\hil_E$ with subsequent action of the quantum channel $\Phi\colon\hil_D\rightarrow\hil_A$ associated with the isometry  $V_1$~(see Fig.~\ref{fig:ges_bip1}). The isometry $V_2$ gets completely traced out and has no effect here. The channel $\Phi$ hence acts on a state $\rho^{\phi}_D=\mathrm{Tr}_E\{\dyad{\phi}_{DE}\}$. Being a convex function, the output norm $\norm{\Phi(\rho^{\phi}_D)}_{\infty}$ attains its maximal value, $\nu_{\infty}(\Phi)$, on pure $\rho^{\phi}_D$~(and, correspondingly, on separable $\ket{\phi}_{DE}$). Consequently, for  the geometric measure of entanglement of the subspace $\mathcal S_{AB_1}\otimes\mathcal G_{B_2C}$ across bipartition $A|BC$ we have
\begin{equation}
    G_{A|BC}(\mathcal S_{AB_1}\otimes\mathcal G_{B_2C}) = 1 - \nu_{\infty}(\Phi).
\end{equation}
On the other hand, the channel $\Phi$ corresponds to the isometry $V_1$ whose range is $\mathcal S_{AB_1}$, and so $\nu_{\infty}(\Phi)$ is equal to the maximum of the first Schmidt coefficient squared taken over all states in $\mathcal S_{AB_1}$. In other words, $G(\mathcal S_{AB_1})=1 - \nu_{\infty}(\Phi)$, and hence
\begin{equation}
   G_{A|BC}(\mathcal S_{AB_1}\otimes\mathcal G_{B_2C}) =  G(\mathcal S_{AB_1}).
\end{equation}

The analysis of bipartition $C|AB$, conducted similarly, yields
\begin{equation}
   G_{C|AB}(\mathcal S_{AB_1}\otimes\mathcal G_{B_2C}) =  G(\mathcal G_{B_2C}).
\end{equation}

\begin{figure}[b]
\includegraphics[scale=0.37]{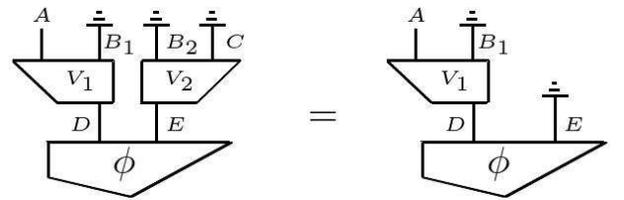}
\caption{\label{fig:ges_bip1}Tracing out subsystems $B$ and $C$ of a state from $\mathcal S_{AB_1}\otimes\mathcal G_{B_2C}$ is equivalent to action of a quantum channel associated to $V_1$ on a state $\rho^{\phi}_D=\mathrm{Tr}_E\{\dyad{\phi}_{DE}\}$.}
\end{figure}

Consider bipartition $B|AC$. Tracing out subsystem $B=B_1B_2$ is equivalent to action of two quantum channels: $\Phi_1\colon\hil_D\rightarrow\hil_A$ and $\Phi_2\colon\hil_E\rightarrow\hil_C$, associated with the isometries $V_1$ and $V_2$ and applied to subsystems $D$ and $E$ of $\ket{\phi}_{DE}$, respectively~(see Fig.~\ref{fig:ges_bip2}).
\begin{figure}[t]
\includegraphics[scale=0.37]{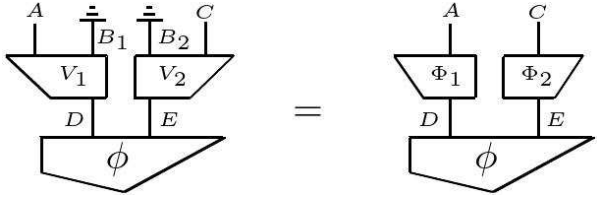}
\caption{\label{fig:ges_bip2}Tracing out subsystem $B$ of a state from $\mathcal S_{AB_1}\otimes\mathcal G_{B_2C}$ is equivalent to action of  quantum channels $\Phi_1$ and $\Phi_2$ on parties $D$ and $E$ of the corresponding state $\ket{\phi}_{DE}$ from $\hil_D\otimes\hil_E$.}
\end{figure}
Analytically this state  can be presented as
\begin{equation}\label{phistate}
    (\Phi_1\otimes\Phi_2)\dyad{\phi}_{DE} = (I\otimes\Phi_2)\,\tau_{DE},
\end{equation}
where $\tau_{DE} = (\Phi_1\otimes I)\dyad{\phi}_{DE}$. For the output norm of this state we have
\begin{equation}\label{geombound}
    \norm{(I\otimes\Phi_2)\,\tau_{DE}}_{\infty}\leqslant\nu(I\otimes\Phi_2)_{\infty} = \nu(\Phi_2)_{\infty},
\end{equation}
where the last equality is due to the property~(\ref{idealcomb}). From Eq.~(\ref{geombound}) it follows that $G_{B|AC}\geqslant G(\mathcal G_{B_2C})$. Actually, $\Phi_1$ and $\Phi_2$ enter Eq.~(\ref{phistate}) symmetrically, and hence another bound for the geometric measure can be written: $G_{B|AC}\geqslant G(\mathcal S_{AB_1})$. Combining the two results, we have:
\begin{equation}\label{bipbound}
G_{B|AC}(\mathcal S_{AB_1}\otimes\mathcal G_{B_2C})\geqslant\max\left(G(\mathcal S_{AB_1}),\,G(\mathcal G_{B_2C})\right).
\end{equation}
Gathering the results across three bipartitions, we obtain Eq.~(\ref{gmeges}).
\end{proof}

\begin{remark}\label{remchan}
The bound in Eq.~(\ref{bipbound}) is not optimal. The geometric measure across bipartition $B|AC$ is directly connected with the maximal output norm of a tensor product of two channels~(as in Eq.~(\ref{phistate})) and the problem of multiplicativity of the maximal output norm,  which was investigated in Refs.~\cite{AmHolWer2000,HolWern02,King03,HaydWin08}. In general, the norm is not multiplicative, and $\nu_p(\Phi_1\otimes\Phi_2)\geqslant\nu_p(\Phi_1)\nu_p(\Phi_2)$. In some particular cases, for example, when one of two channels is entanglement breaking, multiplicativity holds~\cite{King03}. In relation to Lemma~\ref{gesprod} this means that, when one of the completely entangled subspaces in tensor product corresponds to an entanglement breaking channel~(with output purity strictly less than $1$), the geometric measure across bipartition $B|AC$ attains its maximal possible value
\begin{multline}
    G_{B|AC}(\mathcal S_{AB_1}\otimes\mathcal G_{B_2C})=G(\mathcal S_{AB_1}) + G(\mathcal G_{B_2C})\\ - G(\mathcal S_{AB_1})G(\mathcal G_{B_2C}).
\end{multline}
\end{remark}

Lemma~\ref{gesprod} can be extended to the case where $(n+1)$-partite GESs are constructed from tensor product of $n$ bipartite CESs with subsequent joining the adjacent subsystems.
\begin{corollary}
Let $\mathcal{S}^{(1)}_{A_1A_2},\,\mathcal{S}^{(2)}_{A_3A_4},\,\ldots,\,\mathcal{S}^{(n)}_{A_{2n-1}A_{2n}}$  be a system of $n$ bipartite completely entangled subspaces of tensor product Hilbert spaces $\hil_{A_1}\otimes\hil_{A_2},\,\hil_{A_3}\otimes\hil_{A_4},\,\ldots,\,\hil_{A_{2n-1}}\otimes\hil_{A_{2n}}$, respectively~($n\geqslant 2$). Let
 \begin{equation*}
     \mathcal W_{A_1 A'_2A'_3\ldots A'_n A_{2n}}\coloneqq\mathcal{S}^{(1)}_{A_1A_2}\otimes\ldots\otimes\mathcal{S}^{(n)}_{A_{2n-1}A_{2n}}
 \end{equation*}
 be a subspace of an $(n+1)$-partite tensor product Hilbert space $\hil_{A_1}\otimes\hil_{A'_2}\otimes\hil_{A'_3}\otimes\ldots\otimes\hil_{A'_n}\otimes\hil_{A_{2n}}$,
 after taking tensor products and joining subsystems $A_2$ and $A_3$, $A_4$ and $A_5$, \ldots, $A_{2n-2}$ and $A_{2n-1}$ into $A'_2=A_2A_3$, $A'_3=A_4A_5$, \ldots, $A'_n=A_{2n-2}A_{2n-1}$, respectively. Then $\mathcal W_{A_1 A'_2A'_3\ldots A'_n A_{2n}}$ is  genuinely entangled,  with the geometric measure of genuine entanglement 
 \begin{equation}\label{ngeom}
     G_{GME}(\mathcal W_{A_1 A'_2A'_3\ldots A'_n A_{2n}})
     =\min\left(G_1,\,\ldots,\, G_n\right),
 \end{equation}
 where $G_i$ -- the geometric measure of entanglement of the subspace $\mathcal{S}^{(i)}_{A_{2i-1}A_{2i}}$, $1\leqslant i\leqslant n$. 
\end{corollary}
\begin{proof}
In analogy with the proof of Lemma~\ref{gesprod}~(see Fig.~\ref{fig:gesgenod}), the $(n+1)$-partite subspace under consideration is the result of action of  $n$ isometries $\{V_i\}$ on each state $\ket{\phi}$ from an $n$-partite tensor product Hilbert space  $\hil_{C_1}\otimes\ldots\otimes\hil_{C_n}$, with subsequent joining the adjacent subsystems  $A_2$ and $A_3$, \ldots, $A_{2n-2}$ and $A_{2n-1}$ into $A'_2$, \ldots, $A'_n$, respectively~(see Fig.~\ref{fig:nges}).  Here the isometry $V_i$ is associated with the subspace $\mathcal{S}^{(i)}_{A_{2i-1}A_{2i}}$ for $1\leqslant i\leqslant n$.

Now, analyzing entanglement in each of the $2^n-1$ possible bipartite cuts in a way similar to that in the proof of Lemma~\ref{gesprod}, we obtain $2^n-1$ values and lower bounds~(written with the $\geqslant$ signs) for the geometric measure in these cuts:
\begin{multline}\label{allbips}
    G_1,\, G_2,\,\ldots,\,G_n; \geqslant\max(G_1,\,G_2),\,\geqslant\max(G_1,\,G_3),\\\ldots;\,\geqslant\max(G_1,\,G_2,\,G_3),\,\ldots;\,\geqslant\max\left(G_1,\,\ldots,\, G_n\right).
\end{multline}
Here, for example, the value $G_i$ appears in a bipartite cut where, after tracing out appropriate subsystems, only the isometry $V_i$ is left partially traced out, with all other isometries $\{V_k\}$, $k\ne i$ being completely traced out. This situation is analogous to that in the proof of Lemma~\ref{gesprod} shown on Fig.~\ref{fig:ges_bip1}. For another example, the lower bound $\max(G_1,\,G_2,\,G_3)$ appears for a cut where, after tracing out appropriate subsystems, only partially traced out isometries $V_1$, $V_2$, $V_3$ are left, and the rest isometries are completely traced out. This case is analogous to that shown on Fig.~\ref{fig:ges_bip2}~(the difference is that here are three isometries instead of those two presented on the figure).

Noting that the value $\min\left(G_1,\,\ldots,\, G_n\right)$ is the minimum among those in Eq.~(\ref{allbips}), we obtain the equality in Eq.~(\ref{ngeom}).
\end{proof}

\begin{figure}[t]
\includegraphics[scale=0.35]{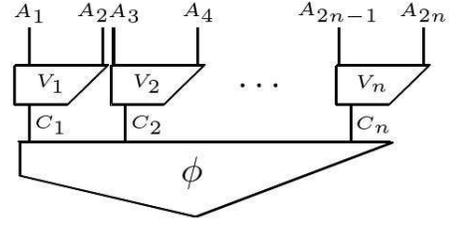}
\caption{\label{fig:nges} The isometry $V_1\otimes\ldots\otimes V_n$ acting on a general  state $\ket{\phi}$ from $\hil_{C_1}\otimes\ldots\otimes\hil_{C_n}$. Action of the isometry on each such state, after joining subsystems $A_2$ and $A_3$, \ldots, $A_{2n-2}$ and $A_{2n-1}$ into $A'_2$, \ldots, $A'_n$ respectively, generates the subspace $\mathcal W_{A_1 A'_2A'_3\ldots A'_n A_{2n}}$, which is a GES.}
\end{figure}

Next we consider some situations where GESs are constructed from direct sums of tensor products of CESs. The following property will be useful here.
\begin{lemma}\label{lemCES}
Let $\mathcal S_{AB_1}$ be a completely entangled  subspace of a tensor product Hilbert space $\hil_A\otimes\hil_{B_1}$. Then the tensor product $\mathcal S_{AB_1}\otimes\hil_{B_2}$, after joining subsystems $B_1$ and $B_2$ into $B=B_1B_2$, is a completely entangled subspace of $\hil_A\otimes\hil_B$.
\end{lemma}
\begin{proof}
Assume that $\mathcal S_{AB_1}$ is spanned by vectors $\ket{\psi_1}_{AB_1},\,\ldots,\,\ket{\psi_n}_{AB_1}$. Let $\ket{\nu_1}_{B_2},\,\ldots,\,\ket{\nu_k}_{B_2}$ be an orthonormal basis in $\hil_{B_2}$.
The elements $\{\ket{\psi_i}_{AB_1}\otimes\ket{\nu_j}_{B_2}\}$ are linearly independent due to  orthonormality of the system $\{\ket{\nu_j}_{B_2}\}$ and linear independence of $\{\ket{\psi_i}_{AB_1}\}$. Let us check that any linear combination of these elements yields an entangled state in $\hil_A\otimes\hil_{B_1B_2}$. If, in some linear combinations, there are elements with the same vector from $\hil_{B_2}$, they can be combined into one term, as in the following example:
\begin{multline}
    c_i\ket{\psi_i}_{AB_1}\otimes\ket{\nu_l}_{B_2} + c_j\ket{\psi_j}_{AB_1}\otimes\ket{\nu_l}_{B_2}  \\
    = c'\ket{\phi}_{AB_1}\otimes\ket{\nu_l}_{B_2},
\end{multline}
where 
\begin{equation*}
    c'\ket{\phi}_{AB_1} = c_i\ket{\psi_i}_{AB_1} + c_j\ket{\psi_j}_{AB_1},
\end{equation*}
with $\ket{\phi}_{AB_1}$ being a normalized state and $c'$ -- some normalization factor. As a linear combination of vectors from a CES, the vector $\ket{\phi}_{AB_1}$ is entangled. Therefore, without loss of generality, one can consider linear combinations
\begin{equation}\label{lincom}
    \sum_{i=1}^k\,c_i\ket{\phi_i}_{AB_1}\otimes\ket{\nu_i}_{B_2},\quad\sum_{i=1}^k\,\abs{c_i}^2 = 1,
\end{equation}
where all the terms have distinct vectors $\{\ket{\nu_i}\}$ from $\hil_{B_2}$, and $\{\ket{\phi_i}_{AB_1}\}$ -- some normalized vectors from the given bipartite CES $\mathcal S_{AB_1}$. Next, tracing out subsystem $B=B_1B_2$ in Eq.~(\ref{lincom}), with the use of the orthonormality property $\mathrm{Tr}_{B_2}\{\ket{\nu_i}\bra{\nu_j}_{B_2}\}=\delta_{ij}$, one obtains the reduced density operator on $\hil_A$:
\begin{multline}\label{mixedreas}
 \sum_{i,\,j=1}^k\,c_i c_j^{*}\,\mathrm{Tr}_{B_1}\{\ket{\phi_i}\bra{\phi_j}_{AB_1}\}\, \mathrm{Tr}_{B_2}\{\ket{\nu_i}\bra{\nu_j}_{B_2}\}\\
 = \sum_{i=1}^k\,\abs{c_i}^2\,\mathrm{Tr}_{B_1}\{\ket{\phi_i}\bra{\phi_i}_{AB_1}\}\\
 \equiv\,\sum_{i=1}^k\,\abs{c_i}^2\,\rho^{\phi_i}_A.
\end{multline}
As a convex sum of mixed states $\rho^{\phi_i}_A$, this state is mixed, and hence the linear combination in Eq.~(\ref{lincom}) yields an entangled state in $\hil_A\otimes\hil_B$.
\end{proof}

The statement of Lemma~\ref{lemCES} can now be slightly changed with the aim to consider direct sums of tensor products.

\begin{corollary}\label{sumprodces}
Let $\mathcal S^{(1)}_{AB_1},\,\ldots,\,\mathcal S^{(n)}_{AB_1}$ be a system of completely entangled subspaces of  $\hil_A\otimes\hil_{B_1}$. Let $\mathcal P^{(1)}_{B_2},\,\ldots,\,\mathcal P^{(n)}_{B_2}$ be a system of mutually orthogonal subspaces of $\hil_{B_2}$. Then the direct sum of tensor products 
\begin{equation}\label{dsum}
    \left(\mathcal S^{(1)}_{AB_1}\otimes\mathcal P^{(1)}_{B_2}\right)\oplus\ldots\oplus\left(\mathcal S^{(n)}_{AB_1}\otimes\mathcal P^{(n)}_{B_2}\right),
\end{equation}
after joining subsystems $B_1$ and $B_2$ into $B=B_1B_2$, is a completely entangled subspace of $\hil_A\otimes\hil_B$.
\end{corollary}
\begin{proof}
Let $\mathcal S^{(r)}_{AB_1}$ be spanned by a system of vectors $\ket{\psi_1^{(r)}}_{AB_1},\,\ldots,\,\ket{\psi_{l_r}^{(r)}}_{AB_1}$ and let $\mathcal P^{(r)}_{AB_1}$ be spanned by an  orthonormal system of vectors  $\ket{\nu_1^{(r)}}_{B_2},\,\ldots,\,\ket{\nu_{k_r}^{(r)}}_{B_2}$, for each $r\colon 1\leqslant r\leqslant n$. An arbitrary vector $\ket{\chi}_{AB}$ that belongs to the direct sum~(\ref{dsum}) can be decomposed as
\begin{equation}\label{lindifces}
 \ket{\chi}_{AB} = \sum_{r=1}^n\sum_{i=1}^{k_r}\, c_i^{(r)}\ket{\phi_i^{(r)}}_{AB_1}\otimes\ket{\nu_i^{(r)}}_{B_2},
 \end{equation}
 where the terms with distinct vectors $\nu$ were gathered and  each $\ket{\phi_i^{(r)}}_{AB_1}$, being a linear combination of $\ket{\psi_1^{(r)}}_{AB_1},\,\ldots,\,\ket{\psi_{l_r}^{(r)}}_{AB_1}$,  is entangled. All vectors $\nu$ are mutually orthogonal: $\bra{\nu_i^{(r)}}\ket{\nu_j^{(s)}}=\delta_{ij}\delta_{rs}$, and hence the linear combination  in Eq.~(\ref{lindifces}) has the same structure as that in Eq.~(\ref{lincom}). Repeating the same reasoning as in Eq.~(\ref{mixedreas}), we obtain that $\ket{\chi}_{AB}$ is entangled.
\end{proof}

\begin{lemma}\label{sumprod}
Let $\mathcal S^{(1)}_{AB_1},\,\ldots,\,\mathcal S^{(n)}_{AB_1}$ be a system of completely entangled subspaces of  $\hil_A\otimes\hil_{B_1}$, and $\mathcal G^{(1)}_{B_2C},\,\ldots,\,\mathcal G^{(n)}_{B_2C}$ -- a system of mutually orthogonal completely entangled subspaces of $\hil_{B_2}\otimes\hil_C$ whose direct sum $\Sigma_{B_2C}\coloneqq\mathcal G^{(1)}_{B_2C}\oplus\,\ldots\oplus\,\mathcal G^{(n)}_{B_2C}$ is also completely entangled. Then the direct sum of tensor products \begin{equation}\label{dsumges}
    \left(\mathcal S^{(1)}_{AB_1}\otimes\mathcal G^{(1)}_{B_2C}\right)\oplus\ldots\oplus\left(\mathcal S^{(n)}_{AB_1}\otimes\mathcal G^{(n)}_{B_2C}\right),
\end{equation}
after joining subsystems $B_1$ and $B_2$ into $B=B_1B_2$, is a genuinely entangled subspace of $\hil_A\otimes\hil_B\otimes\hil_C$.
\end{lemma}
\begin{proof}
Let $\hil_{B_2}$ be a Hilbert space of dimension equal to the dimension of $\Sigma_{B_2C}$. Consider an isometry $V\colon\,\hil_{B_2}\rightarrow\hil_{B_2}\otimes\hil_C$ that maps $\hil_{B_2}$ to $\Sigma_{B_2C}$.  The isometry has a CES as its range, and so it corresponds to a quantum channel with output purity strictly less than $1$. By Eq.~(\ref{idealcomb}), so does the isometry $I_{\scriptscriptstyle B_1}\otimes V_{\scriptscriptstyle B_2\rightarrow B_2C}$. 

Let $\mathcal P^{(1)}_{B_2},\,\ldots,\,\mathcal P^{(n)}_{B_2}$ be a system of mutually orthogonal subspaces of $\hil_{B_2}$. By Corollary~\ref{sumprodces}, 
\begin{equation*}
    \Omega_{AB}\coloneqq\left(\mathcal S^{(1)}_{AB_1}\otimes\mathcal P^{(1)}_{B_2}\right)\oplus\ldots\oplus\left(\mathcal S^{(n)}_{AB_1}\otimes\mathcal P^{(n)}_{B_2}\right)
\end{equation*}
is a CES of $\hil_A\otimes\hil_B$. The subspace in Eq.~(\ref{dsumges}) is obtained from $\Omega_{AB}$ by the action of the isometry $I_{\scriptscriptstyle B_1}\otimes V_{\scriptscriptstyle B_2\rightarrow B_2C}$ on subsystem $B=B_1B_2$~(see Fig.~\ref{fig:gesgen}). This situation corresponds to the scheme on Fig.~\ref{fig:iso23}. Therefore, the generated subspace  is genuinely entangled.
\end{proof}

\begin{figure}[t]
\includegraphics[scale=0.35]{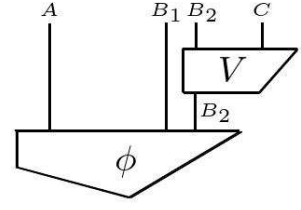}
\caption{\label{fig:gesgen} Action of the isometry $I_{\scriptscriptstyle B_1}\otimes V_{\scriptscriptstyle B_2\rightarrow B_2C}$ on each state from $\Omega_{AB}$  generates the GES presented in Eq.~(\ref{dsumges}).}
\end{figure}

Note that the CESs $S^{(1)}_{AB_1},\,\ldots,\,\mathcal S^{(n)}_{AB_1}$ in the above statement can be arbitrary, and they can have arbitrary relations to each other~(e.~g, intersect or not intersect). In particular, each of them can be spanned by just one entangled vector.

\begin{corollary}\label{corent}
Let $\ket{\psi_1}_{AB_1},\,\ldots,\,\ket{\psi_n}_{AB_1}$ be some entangled vectors in $\hil_A\otimes\hil_{B_1}$, and $\ket{\chi_1}_{B_2C},\,\ldots,\,\ket{\chi_n}_{B_2C}$ -- mutually orthogonal  vectors spanning a completely entangled subspace of $\hil_{B_2}\otimes\hil_C$. Then a system of vectors $\ket{\psi_1}_{AB_1}\otimes\ket{\chi_1}_{B_2C},\,\ldots,\,\ket{\psi_n}_{AB_1}\otimes\ket{\chi_n}_{B_2C}$ spans a genuinely entangled subspace of $\hil_A\otimes\hil_B\otimes\hil_C$.
\end{corollary}

\section{Applications}\label{sec::app}

The established properties can have several  applications. 

\subsection{Tensor products of mixed bipartite entangled states}

In Refs.~\cite{ShenChen20,SunChen21} it was stated as a conjecture that a tensor product of two mixed  bipartite entangled states, $\alpha_{AB_1}\otimes\beta_{B_2C}$, after joining $B_1$ and $B_2$, is a genuinely entangled tripartite state. Later the conjecture was disproved in Ref.~\cite{TenNet2021} by finding an example with two entangled isotropic states whose tensor product is not GE.
In this connection, it is interesting to search for sufficient conditions of genuine entanglement of such tensor products.

One condition of this type can be obtained from combining the properties of  tensor products of CESs with a particular  witness of genuine entanglement connected with projection on some GES, namely, in Ref.~\cite{KVAnt21} it was shown that if, for a multipartite state $\rho$ and a genuinely entangled subspace $W$, the inequality
\begin{equation}\label{EntWit}
    \Tr{\rho\,\Pi_W} + G_{GME}(W) - 1 >0
\end{equation}
holds, then $\rho$ is genuinely entangled. Here $\Pi_W$ -- an orthogonal projector onto $W$.
\begin{lemma}\label{StProd}
Let $\alpha_{AB_1}$ and $\beta_{B_2C}$ be  two bipartite mixed states on $\hil_{A}\otimes\hil_{B_1}$ and $\hil_{B_2}\otimes\hil_{C}$, respectively. Let $W_1$ and $W_2$ be two  completely entangled subspaces of  $\hil_{A}\otimes\hil_{B_1}$ and $\hil_{B_2}\otimes\hil_{C}$, respectively. Then the tensor product $\alpha_{AB_1}\otimes\beta_{B_2C}$, after joining $B_1$ and $B_2$ into $B=B_1B_2$, is a genuinely entangled tripartite state on $\hil_A\otimes\hil_B\otimes\hil_C$ if
\begin{multline}\label{EnCon}
    \Tr{\alpha_{AB_1}\,\Pi_{W_1}}\Tr{\beta_{B_2C}\,\Pi_{W_2}}  \\ > 1 - \min\left(G(W_1),\,G(W_2)\right).
\end{multline}
\end{lemma}
\begin{proof}
We can use condition~(\ref{EntWit}) with respect to the state $\alpha_{AB_1}\otimes\beta_{B_2C}$ and the subspace $W_1\otimes W_2$ of the tensor product Hilbert space $\hil_A\otimes\hil_B\otimes\hil_C$. By Lemma~\ref{gesprod}, $W_1\otimes W_2$ is a GES, with the GME geometric measure 
\begin{equation}\label{minG}
    G_{GME}(W_1\otimes W_2) = \min\left(G(W_1),\,G(W_2)\right).
\end{equation}
In addition,
\begin{multline}\label{TrF}
\Tr{\alpha_{AB_1}\otimes\beta_{B_2C}\,\Pi_{W_1\otimes W_2}}  \\ =
\Tr{\alpha_{AB_1}\,\Pi_{W_1}}\Tr{\beta_{B_2C}\,\Pi_{W_2}}.
\end{multline}
Combining Eqs.~(\ref{EntWit}), (\ref{minG}), and (\ref{TrF}), we obtain  sufficient condition~(\ref{EnCon}) for genuine entanglement  of $\alpha_{AB_1}\otimes\beta_{B_2C}$.
\end{proof}
\begin{remark}
Lower bounds on two GME entanglement measures, the concurrence and the convex-roof extended negativity~(CREN), can be also obtained in connection with this entanglement witness. For example, if condition~(\ref{EnCon}) holds,  Eq.~(64) from Ref.~\cite{KVAnt21} yields the bound for the CREN of the state $\alpha_{AB_1}\otimes\beta_{B_2C}$:
\begin{multline}\label{CREN}
    N_{GME}(\alpha_{AB_1}\otimes\beta_{B_2C})\\ \geqslant\frac{\Tr{\alpha_{AB_1}\,\Pi_{W_1}}\Tr{\beta_{B_2C}\,\Pi_{W_2}} + G_{12} - 1}{2\,(1 - G_{12})},
\end{multline}
where $G_{12}=\min\left(G(W_1),\,G(W_2)\right)$.
\end{remark}

\subsubsection*{Example: tensor product of two Werner states}

Consider the  Werner states family  on $\mathbb{C}^d\otimes\mathbb{C}^d$:
\begin{equation}
    \rho_{\scriptscriptstyle\mathcal W}(p,d) = \frac1{d^2 + pd}\left(I_d\otimes I_d + p\sum_{i,\,j=0}^{d-1}\,\ket{i,\,j}\bra{j,\,i}\right).
\end{equation}

In Ref.~\cite{SunChen21} it was proved that $\rho_{\scriptscriptstyle\mathcal W}(p_1,2)\otimes\rho_{\scriptscriptstyle\mathcal W}(p_2,2)$, when viewed as a tripartite state on  $\mathbb{C}^2\otimes\mathbb{C}^4\otimes\mathbb{C}^2$, is genuinely entangled in the region
\begin{equation}\label{oldDom}
    -1\,\leqslant\,p_1\,\leqslant -0.940198;\quad -1\,\leqslant\,p_2\,\leqslant\,-0.94066.
\end{equation}
With the use of Lemma~\ref{StProd} this domain can be extended.

Let us consider the tensor product $\rho_{\scriptscriptstyle\mathcal W}(p_1,d)\otimes\rho_{\scriptscriptstyle\mathcal W}(p_2,d)$ of two Werner states on $\mathbb{C}^d\otimes\mathbb{C}^d$. With the use of relations
\begin{equation}
    \Pi_{\mathcal A} = \frac{I-\mathrm{SWAP}}2;\quad\Pi_S=\frac{I+\mathrm{SWAP}}2,
\end{equation}
where $\Pi_{\mathcal A},\,\Pi_S$ -- the projectors onto the antisymmetric and the symmetric subspaces of $\mathbb{C}^d\otimes\mathbb{C}^d$ respectively, and  
\begin{equation}
    \sw = \sum_{i,\,j=0}^{d-1}\,\ket{i,\,j}\bra{j,\,i},
\end{equation}
the  operator that exchanges qudits, the Werner state itself can be rewritten as
\begin{equation}
    \rho_{\scriptscriptstyle\mathcal W}(p,d) = \frac1{d^2+pd}\left[(1+p)\Pi_S + (1-p)\Pi_{\mathcal A}\right].
\end{equation}
For our analysis it is more convenient to reparameterize it with a new variable $s$ related to $p$ as
\begin{equation}\label{sp}
    \frac{2s}{d(d-1)} = \frac{1-p}{d(p+d)},
\end{equation}
so  that
\begin{equation}\label{NpW}
   \rho_{\scriptscriptstyle\mathcal W}(s,d) = \frac{2(1-s)}{d(d+1)}\Pi_S + \frac{2s}{d(d-1)}\Pi_{\mathcal A}.
\end{equation}
Let us apply Lemma~\ref{StProd} and condition~(\ref{EnCon}) to the state $\rho_{\scriptscriptstyle\mathcal W}(s_1,d)\otimes\rho_{\scriptscriptstyle\mathcal W}(s_2,d)$, with both $W_1$ and $W_2$ chosen to be the antisymmetric subspace $\mathcal A$ of  $\mathbb{C}^d\otimes\mathbb{C}^d$, which has dimension equal to $d(d-1)/2$. It is known~\cite{Vid02} that the geometric measure $G(\mathcal{A}) = 1/2$~(see also \cite{KVAnt20}). From Eq.~(\ref{NpW}) it follows that $$\Tr{\rho_{\scriptscriptstyle\mathcal W}(s,d)\,\Pi_{\mathcal A}} = s,$$
and thus condition~(\ref{EnCon}) takes a simple form:
\begin{equation}\label{domEn}
    s_1 s_2 > \frac12.
\end{equation}
\begin{figure}[t]
\includegraphics[scale=0.55]{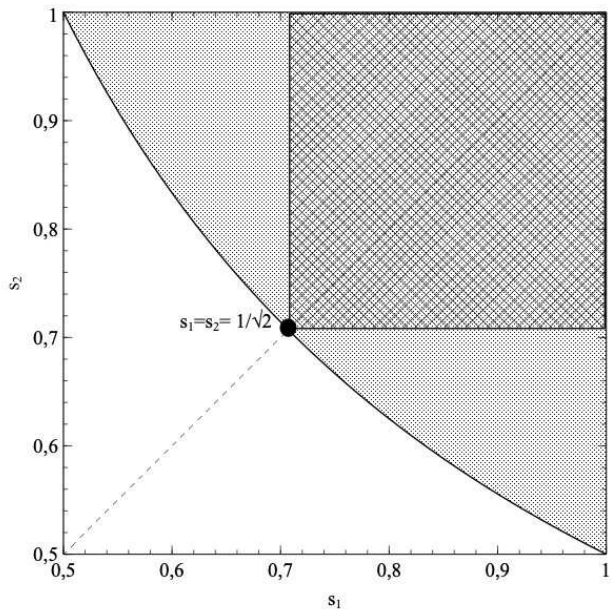}
\caption{\label{fig:plot} Shaded regions represent the domains in $s_1,\,s_2$ where the state $\rho_{\scriptscriptstyle\mathcal W}(s_1,d)\otimes\rho_{\scriptscriptstyle\mathcal W}(s_2,d)$ is genuinely entangled: the area above the graph $s_2=1/(2s_1)$ and its maximal square subdomain.}
\end{figure}
Eq.~(\ref{domEn}) defines the domain of genuine entanglement of the state $\rho_{\scriptscriptstyle\mathcal W}(s_1,d)\otimes\rho_{\scriptscriptstyle\mathcal W}(s_2,d)$~(the whole shaded area above the graph $s_2 = 1/(2s_1)$ depicted on Fig.~\ref{fig:plot}). In particular, we can specify the maximal square subdomain~(also shown on Fig.~\ref{fig:plot})
\begin{equation}
    \frac1{\sqrt 2} < s_1 \leqslant 1;\quad \frac1{\sqrt 2} < s_2 \leqslant 1,
\end{equation}
where the parameters $s_1$ and $s_2$ vary independently and where this state is GE. With the use of Eq.~(\ref{sp}), this region can be rewritten in terms of $p_1,\,p_2$. For $d=2$ we obtain
\begin{equation}
    -1\leqslant p_1,\,p_2 < 3\sqrt2 -5\approx -0.757359,
\end{equation}
which extends the domain in Eq.~(\ref{oldDom}). For larger $d$ the region becomes even wider:
\begin{equation}
  -1\leqslant p_1,\,p_2 < \frac{d(1-\sqrt2)-1}{\sqrt2 +d -1},  
\end{equation}
with the upper bound tending to $1-\sqrt2\approx -0.414213$, when $d\rightarrow\infty$.

In addition,  Eq.~(\ref{CREN}) yields a lower bound on the negativity of the state, which reads as
\begin{equation}
    N_{GME}(\rho_{\scriptscriptstyle\mathcal W}(s_1,d)\otimes\rho_{\scriptscriptstyle\mathcal W}(s_2,d)) \geqslant s_1s_2 - \frac12,
\end{equation}
or, by Eq.~(\ref{sp}),
\begin{multline}
    N_{GME}(\rho_{\scriptscriptstyle\mathcal W}(p_1,d)\otimes\rho_{\scriptscriptstyle\mathcal W}(p_2,d)) \\\geqslant
    \frac{(d-1)^2(1-p_1)(1-p_2)}{4(p_1+d)(p_2+d)} - \frac12.
\end{multline}

\subsection{Construction of multipartite NPT and distillable subspaces}

An important aspect in the tasks of quantum information processing is the possibility to extract pure entangled states from mixed ones. The states from which pure entanglement can be obtained are called distillable~\cite{BenVinSmWoot96}.

More formally, a state $\rho$ on $\hil_A\otimes\hil_B$ is $1$-distillable~(or one-copy distillable)~\cite{VSSTT00} if there exists a pure Schmidt rank $2$ bipartite state $\ket{\psi}$ such that 
\begin{equation}
    \bra{\psi}\rho^{T_A}\ket{\psi} < 0,
\end{equation}
where $T_A$ -- the transpose operation applied on subsystem $A$~(the partial transpose). Next, a state $\rho$ is $n$-distillable if $\rho^{\otimes n}$ is $1$-distillable. 

All distillable states are necessarily NPT - those  with partial transpose having at least one negative eigenvalue~(non-positive partial transpose). It is an open question whether the converse is true.

A multipartite subspace is called NPT with respect to some bipartite cut if any density operator with support in the subspace is NPT across this bipartite cut.
 Such subspaces can serve as a source of various mixed NPT states that could potentially be distillable. There are several known constructions of multipartite subspaces that are NPT with respect to certain bipartite cuts~\cite{SAS14,JLP19}. In particular, Ref.~\cite{JLP19} provides the method of construction of maximal multipartite subspaces  that are NPT across at least one bipartite cut. 
 
 In this subsection we show that  $(n+1)$-partite  subspaces that are NPT with respect to \emph{any} bipartite cut can be constructed from $n$ bipartite NPT subspaces. We call a multipartite subspace $1$-distillable across some bipartite cut if any density operator supported on the subspace is $1$-distillable across this cut.
 \begin{lemma}\label{sprodnpt}
 Let $\mathcal{S}^{(1)}_{A_1A_2},\,\mathcal{S}^{(2)}_{A_3A_4},\,\ldots,\,\mathcal{S}^{(n)}_{A_{2n-1}A_{2n}}$  be a system of $n$ bipartite NPT subspaces of tensor product Hilbert spaces $\hil_{A_1}\otimes\hil_{A_2},\,\hil_{A_3}\otimes\hil_{A_4},\,\ldots,\,\hil_{A_{2n-1}}\otimes\hil_{A_{2n}}$, respectively~($n\geqslant 2$). Let
 \begin{equation*}
     \mathcal W_{A_1 A'_2A'_3\ldots A'_n A_{2n}}\coloneqq\mathcal{S}^{(1)}_{A_1A_2}\otimes\ldots\otimes\mathcal{S}^{(n)}_{A_{2n-1}A_{2n}}
 \end{equation*}
 be a subspace of an $(n+1)$-partite tensor product Hilbert space $\hil_{A_1}\otimes\hil_{A'_2}\otimes\hil_{A'_3}\otimes\ldots\otimes\hil_{A'_n}\otimes\hil_{A_{2n}}$,
 after taking tensor products and joining subsystems $A_2$ and $A_3$, $A_4$ and $A_5$, \ldots, $A_{2n-2}$ and $A_{2n-1}$ into $A'_2=A_2A_3$, $A'_3=A_4A_5$, \ldots, $A'_n=A_{2n-2}A_{2n-1}$, respectively. Then $\mathcal W_{A_1 A'_2A'_3\ldots A'_n A_{2n}}$ is NPT across any bipartite cut. If, in addition, each of bipartite subspaces $\mathcal S$ is $1$-distillable, then $\mathcal W_{A_1 A'_2A'_3\ldots A'_n A_{2n}}$ is $1$-distillable across any bipartite cut.
 \end{lemma}
 See Appendix~\ref{app:lemp} for the proof.

\subsubsection*{Example: construction of a tripartite subspace $1$-distillable across any bipartite cut}
 We construct this example from tensor product of two $1$-distillable bipartite subspaces. To find such bipartite subspaces, we use the argument from Ref.~\cite{AgHalBa19} which combines the results of Refs.~\cite{NJohn13,ChenDjok16}. In Ref.~\cite{NJohn13} it was shown that for a bipartite $\mathbb C^{d_1}\otimes\mathbb C^{d_2}$ system  NPT  subspaces of dimension up to $(d_1-1)(d_2-1)$ can be constructed. The NPT subspace $\mathcal S$ of maximal dimension reads as
 \begin{multline}\label{nptex}
     \mathcal S\coloneqq\mathrm{span}\{\ket{j}\ket{k+1} - \ket{j+1}\ket{k}\},\\
     0\leqslant j\leqslant d_1-2,\quad 0\leqslant k\leqslant d_2-2.
 \end{multline}
 (Theorem~1 of Ref.~\cite{NJohn13}).
 
 Next, in Ref.~\cite{ChenDjok16} it was shown that any rank 4 NPT state is $1$-distillable, which, combined with the results of Ref.~\cite{ChenDjok11}, means that all NPT states of rank \emph{at most} 4 are $1$-distillable. Therefore, bipartite subspaces~(\ref{nptex}) of dimensions up to $4$  are $1$-distillable.
 
 Using the above facts, we can take a subspace $\mathcal S$ of Eq.~(\ref{nptex})  with $d_1=d_2=3$, such that $\dim(\mathcal S)=4$.  Let $\mathcal W$ denote a subspace obtained from the tensor product $\mathcal S\otimes\mathcal S$ of $\mathcal S$ with itself, with subsequent joining the two adjacent subsystems. $\mathcal W$ is hence a $16$-dimensional subspace of a tripartite $3\otimes 9\otimes 3$ Hilbert space. According to Lemma~\ref{sprodnpt}, $\mathcal W$ is $1$-distillable across any of the three bipartite cuts. 
 
 The subspace $\mathcal W$ is spanned by the system of $16$ vectors obtained from all possible tensor products of vectors from $\mathcal S$ with each other. After taking the tensor products the two adjacent subsystems are to be joined according to the lexicographic order:
 \begin{multline}
     \ket{0}\ket{0}\rightarrow \ket{0},\quad\ket{0}\ket{1}\rightarrow\ket{1},\\
     \ldots,\quad\ket{2}\ket{2}\rightarrow\ket{8},
 \end{multline}
 or, more generally,
 \begin{equation}\label{order}
    \ket{i}\ket{j}\rightarrow \ket{3i+j}.
 \end{equation}
The tensor product of two vectors from~(\ref{nptex}) 
\begin{equation}
    \left(\ket{j}\ket{k+1} - \ket{j+1}\ket{k}\right)\otimes\left(\ket{l}\ket{m+1} - \ket{l+1}\ket{m}\right),
\end{equation}
indexed by $(j,\,k)$ and $(l,\,m)$ respectively,
yields, by Eq.~(\ref{order}),  a generic vector from the system of vectors spanning $\mathcal W$:
\begin{multline}
    \ket{j}\ket{3(k+1)+l}\ket{m+1}-\ket{j}\ket{3(k+1)+l+1}\ket{m}\\
    -\ket{j+1}\ket{3k+l}\ket{m+1}+\ket{j+1}\ket{3k+l+1}\ket{m},\\
    0\leqslant j,\,k,\,l,\,m\leqslant1.
\end{multline}

\subsection{Entanglement criterion}

Corollary~\ref{corent} can be combined with some known results to give entanglement conditions for  mixed states supported on tensor products. We give one such example using the result of Ref.~\cite{DRA21}, a simple sufficient condition for a subspace to be completely entangled:
\begin{theorem}[Ref.~\cite{DRA21}]\label{exttheor}
Let $V$ be a subspace spanned by $k$ pairwise orthogonal pure bipartite states $\{\ket{\phi_i}\}$ such that
\begin{equation}
    \sum_{i=1}^k\,G(\ket{\psi_i}) - (k-1) > 0,
\end{equation}
where $G$ -- the geometric measure of entanglement. Then $V$ is a completely entangled subspace.
\end{theorem}

Combining it with Corollary~\ref{corent}, we obtain some sort of an entanglement criterion.

\begin{lemma}
Let $\rho =\sum_{i=1}^n\,\dyad{\psi_i}$ be a density operator on a tripartite tensor product Hilbert space $\hil_A\otimes\hil_B\otimes\hil_C$, where each state $\ket{\psi_i}$ is obtained from tensor product $\ket{\phi_i}_{AB_1}\otimes\ket{\chi_i}_{B_2C}$ of pure states $\ket{\phi_i}_{AB_1}\in\hil_A\otimes\hil_{B_1}$ and $\ket{\chi_i}_{B_2C}\in\hil_{B_2}\otimes\hil_C$, with subsequent joining  subsystems $B_1$ and $B_2$ into $B$. Suppose that each $\ket{\phi_i}_{AB_1}$ is entangled. Suppose that  $\{\ket{\chi_i}_{B_2C}\}$ are mutually orthogonal and such that
\begin{equation*}
    \sum_{i=1}^n\,G(\ket{\chi_i}) - (n-1) > 0.
\end{equation*}
Then $\rho$ is a genuinely entangled state.
\end{lemma}
\begin{proof}
By Corollary~\ref{corent} and Theorem~\ref{exttheor} the states $\{\ket{\psi_i}\}$ span a GES. As a state supported on a GES, $\rho$ is genuinely entangled.
\end{proof}

\section{Discussion}\label{sec::disc}

We have presented several properties of genuinely entangled subspaces obtained from the tensor product structure.

The advantage of such a construction is the possibility to control such useful characteristics of states supported on the output GESs as various measures of entanglement, distillability across some or all bipartite cuts, robustness of entanglement under mixing with external noise~(not covered here, but it easily follows from Eqs.~(68)-(71) of Ref.~\cite{KVAnt21}). In particular, highly entangled subspaces can be generated in this way. In addition, if a tripartite GES is constructed from two CESs with given geometric measures of entanglement, and one of them corresponds to an entanglement breaking channel, then, according to Remark on page~\pageref{remchan}, the exact values of the geometric  measure across all three bipartite cuts are known for the resulting GES.

It has also been shown that, under certain conditions, GESs can be obtained from the direct sum of tensor products of bipartite CESs~(Lemma~\ref{sumprod}). Such a structure reminds of the inner product of vectors in the Euclidean space, although here in Lemma~\ref{sumprod} the conditions are not symmetric with respect to the left and the right subspaces in tensor products. In addition, as it was shown in Ref.~\cite{KVAnt21}, the scheme of Fig.~\ref{fig:iso23}, used in the proof of the lemma, cannot generate GESs of maximal possible dimensions, although the dimensions of output GESs  asymptotically approach the maximal ones when local dimensions of subsystems are high. Therefore, the construction of Lemma~\ref{sumprod} doesn't generate maximal GESs either. A possible direction of further research can be the generalization of Lemma~\ref{sumprod} with the aim to obtain more symmetric conditions on bipartite subspaces as well as conditions sufficient for construction of maximal GESs.

\begin{acknowledgments}
The author thanks M. V. Lomonosov Moscow State University for supporting this work.
\end{acknowledgments}

\appendix

\section{\label{app:lemp}Proof of Lemma~\ref{sprodnpt}}

\begin{proof}
 We prove the lemma for $n=2$, the case of arbitrary $n$ can be considered in a similar way.
 
 Let $\rho$ be a density operator  supported on 
 $\mathcal W_{A_1 A'_2A_4}=\mathcal{S}^{(1)}_{A_1A_2}\otimes\mathcal{S}^{(2)}_{A_3A_4}$, where $A'_2=A_2A_3$~(we use $A'_2$ and $A_2A_3$ interchangeably), so that $\rho$ has an ensemble decomposition
 \begin{equation}\label{dec}
     \rho = \sum_i\,p_i\,\dyad{\psi_i}_{A_1A'_2A_4},
 \end{equation}
 with $\ket{\psi_i}_{A_1A'_2A_4}$ being decomposed as
 \begin{equation}\label{vecdec}
    \ket{\psi_i}_{A_1A'_2A_4} = \sum_{jk}\,c^{(i)}_{jk}\,\ket{\phi_j}_{A_1A_2}\otimes\ket{\chi_k}_{A_3A_4}, 
 \end{equation}
 where $\ket{\phi_j}_{A_1A_2}\in\mathcal{S}^{(1)}_{A_1A_2}$, $\ket{\chi_k}_{A_3A_4}\in\mathcal{S}^{(2)}_{A_3A_4}$, and $c^{(i)}_{jk}\in\mathbb C$.
 
 For the bipartite cut $A_1|A'_2A_4$ we choose the partial transpose to act on subsystem $A_1$. We want to show that there is a pure state $\ket{\Gamma}\in\hil_{A_1}\otimes\hil_{A'_2}\otimes\hil_{A_4}$ such that
 \begin{equation}\label{npt}
     \bra{\Gamma}\rho^{T_{A_1}}\ket{\Gamma} < 0.
 \end{equation}
 We can take $\ket{\Gamma}$ to have structure
 \begin{equation}\label{gammarepr}
     \ket{\Gamma}_{A_1A_2A_3A_4} = \ket{\Phi}_{A_1A_2}\otimes\ket{\tau}_{A_3A_4},
 \end{equation}
 (before joining $A_2$ and $A_3$), with some pure states $\ket{\Phi}_{A_1A_2}\in\hil_{A_1}\otimes\hil_{A_2}$, $\ket{\tau}_{A_3A_4}\in\hil_{A_3}\otimes\hil_{A_4}$. Now, for each term in decomposition~(\ref{dec}), it can be noted that  in expression
 \begin{multline}\label{transfppt}
     \bra{\Gamma}\left(\dyad{\psi_i}_{A_1A'_2A_4}\right)^{T_{A_1}}\ket{\Gamma}\\
     =\bra{\Phi}\otimes\bra{\tau}\left(\dyad{\psi_i}_{A_1A'_2A_4}\right)^{T_{A_1}}\ket{\Phi}_{A_1A_2}\otimes\ket{\tau}_{A_3A_4}
 \end{multline}
 the operations $T_{A_1}$ and scalar product with $\ket{\tau}_{A_3A_4}$ can be taken independently~(as acting on different subsystems). So, first  taking a partial scalar product of $\ket{\psi_j}$ with $\ket{\tau}$, with the use of Eq.~(\ref{vecdec}) we obtain
 \begin{multline}
     \bra{\tau}_{A_3A_4}\ket{\psi_i}_{A_1A'_2A_4}=\sum_{jk}\,c^{(i)}_{jk}\,\ket{\phi_j}_{A_1A_2} \bra{\tau}\ket{\chi_k}_{A_3A_4}\\
     =\sum_j\,\Tilde{c}^{(i)}_j\ket{\phi_j}_{A_1A_2}=n_i\ket{\eta_i}_{A_1A_2},
 \end{multline}
 where $\ket{\eta_i}_{A_1A_2}$ -- some normalized state from the subspace $\mathcal{S}^{(1)}_{A_1A_2}$ and $n_i>0$ -- the corresponding normalization constant. Now the left part of Eq.~(\ref{npt}) can be written as
 \begin{equation}
     \bra{\Gamma}\rho^{T_{A_1}}\ket{\Gamma} = c\bra{\Phi}\sigma^{T_{A_1}}\ket{\Phi}_{A_1A_2},
\end{equation}
where
\begin{equation}\label{sig}
    \sigma = \sum_i\,\Tilde{p_i}\dyad{\eta_i}_{A_1A_2},
\end{equation}
a state entirely supported on $\mathcal{S}^{(1)}_{A_1A_2}$, with
\begin{equation}
    \Tilde{p_i}=\frac{n_i^2 p_i}c,\quad c = \sum_i\,n_i^2 p_i.
\end{equation}
Since the state $\sigma$ is NPT, choosing in Eq.~(\ref{gammarepr}) the state $\Phi$ such that
\begin{equation}\label{partnpt}
\bra{\Phi}\sigma^{T_{A_1}}\ket{\Phi}_{A_1A_2}<0,
\end{equation}
we obtain the state $\ket{\Gamma}$ for which condition~(\ref{npt}) is satisfied, and this shows that $\mathcal W_{A_1 A'_2A_4}$ is NPT across bipartite cut $A_1|A'_2A_4$.
\begin{figure}[t]
\includegraphics[scale=0.25]{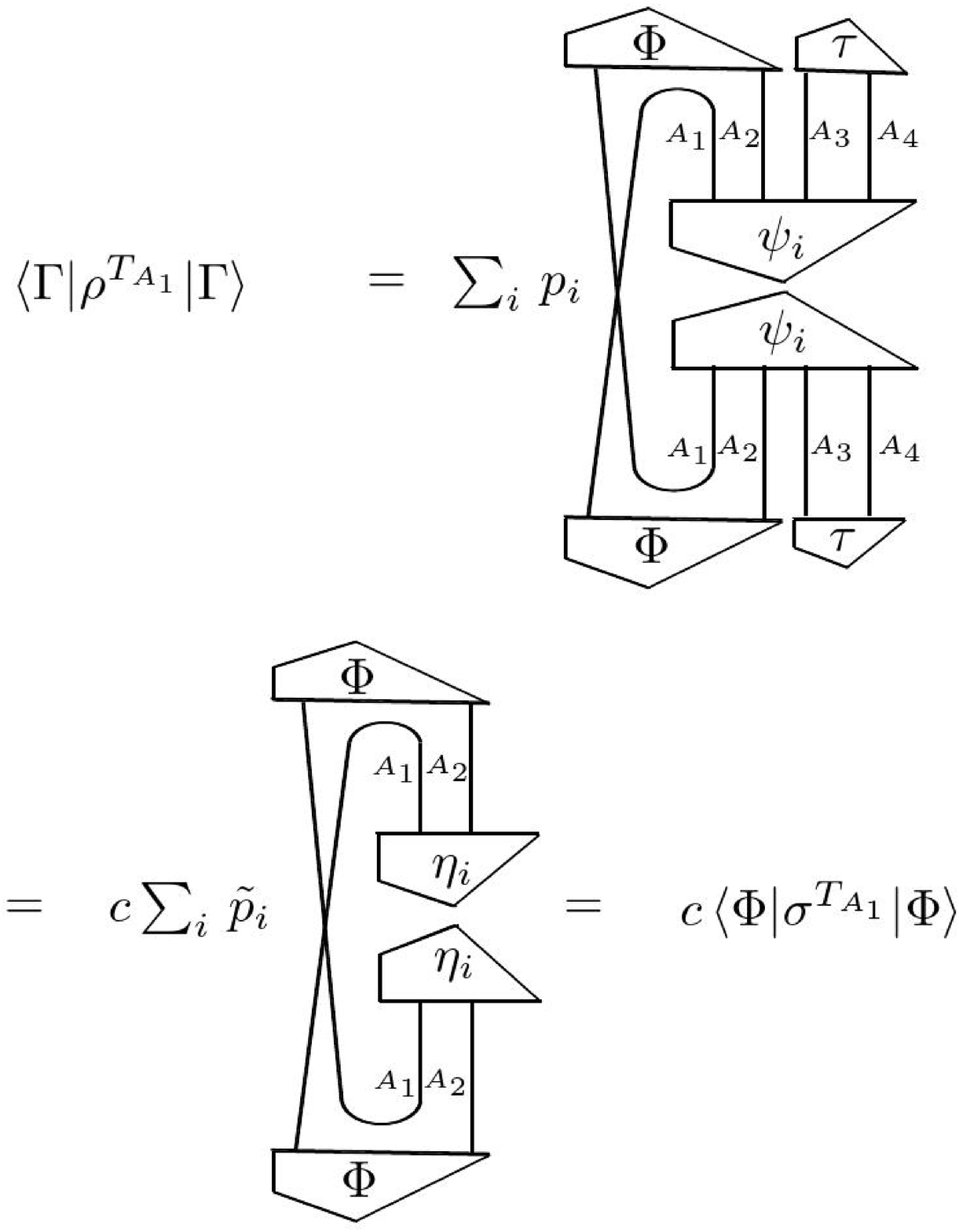}
\caption{\label{fig:ppt1} The operations of partial transpose $T_{A_1}$~(depicted by bending of lines)  of each projector $\dyad{\psi_i}$ and partial scalar product of $\ket{\tau}$ with $\ket{\psi_i}$ are performed independently. The resulting states from partial scalar product, $\{\ket{\eta_i}\}$, belong to $\mathcal{S}^{(1)}_{A_1A_2}$, and the density operator $\sigma=\sum_i\,\Tilde{p_i}\dyad{\eta_i}$ is entirely supported on $\mathcal{S}^{(1)}_{A_1A_2}$.}
\end{figure}

The reasoning in Eqs.~(\ref{transfppt})-(\ref{sig}) can be conveniently represented diagrammatically, as shown on Fig.~\ref{fig:ppt1}.

If, in addition, subspace $\mathcal{S}^{(1)}_{A_1A_2}$ is $1$-distillable, then there exists a Schmidt rank $2$ state $\ket{\Phi}_{A_1A_2}$ such that condition~(\ref{partnpt}) is satisfied. Using this state in Eq.~(\ref{gammarepr}), we construct a Schmidt rank $2$ state $\ket{\Gamma}$~(again, after joining $A_2$ and $A_3$) such that condition~(\ref{npt}) is satisfied, thus proving $1$-distillability of $\mathcal W_{A_1 A'_2A_4}$ across bipartite cut $A_1|A'_2A_4$. 

The same holds for bipartite cut $A_1A'_2|A_4$~(subspaces $\mathcal{S}^{(1)}_{A_1A_2}$ and ${S}^{(2)}_{A_3A_4}$ enter the lemma symmetrically).

Consider now bipartite cut $A'_2|A_1A_4$. This time we choose the partial transpose to act on joint subsystem $A_1A_4$. This operation reduces to taking transposes on subsystems $A_1$ and $A_4$ independently: $T_{A_1A_4}=T_{A_1}\otimes T_{A_4}$.

For the state $\ket{\Gamma}$ we can take the structure~(\ref{gammarepr}) requiring the state $\ket{\tau}_{A_3A_4}$ to be a product state: 
\begin{equation}\label{prodstr}
    \ket{\tau}_{A_3A_4} = \ket{\mu}_{A_3}\otimes\ket{\nu}_{A_4},
\end{equation}
with some pure states $\ket{\mu}_{A_3}\in\hil_{A_3}$ and $\ket{\nu}_{A_4}\in\hil_{A_4}$.

Now, for each term in Eq.~(\ref{dec}), the partial scalar product of $\ket{\tau}$ with the transposed projector $\dyad{\psi_i}$ can be written as
\begin{multline}
\bra{\tau}\left(\dyad{\psi_i}_{A_1A'_2A_4}\right)^{T_{A_1}\otimes T_{A_4}}\ket{\tau}_{A_3A_4}\\
=\bra{\mu}\otimes\bra{\nu}\left(\dyad{\psi_i}_{A_1A'_2A_4}\right)^{T_{A_1}\otimes T_{A_4}}\ket{\mu}_{A_3}\otimes\ket{\nu}_{A_4}\\
=\bra{\mu}\otimes\bra{\nu^*}\left(\dyad{\psi_i}_{A_1A'_2A_4}\right)^{T_{A_1}}\ket{\mu}_{A_3}\otimes\ket{\nu^*}_{A_4},
\end{multline}
where we took advantage of the product structure~(\ref{prodstr}) to eliminate the second transpose operation $T_{A_4}$~(see also Fig.~\ref{fig:ppt2}). Here $\ket{\nu^*}$ denotes the vector with components equal to complex conjugated components of the vector $\ket{\nu}$ with respect to the computational basis. 
\begin{figure}[b]
\includegraphics[scale=0.25]{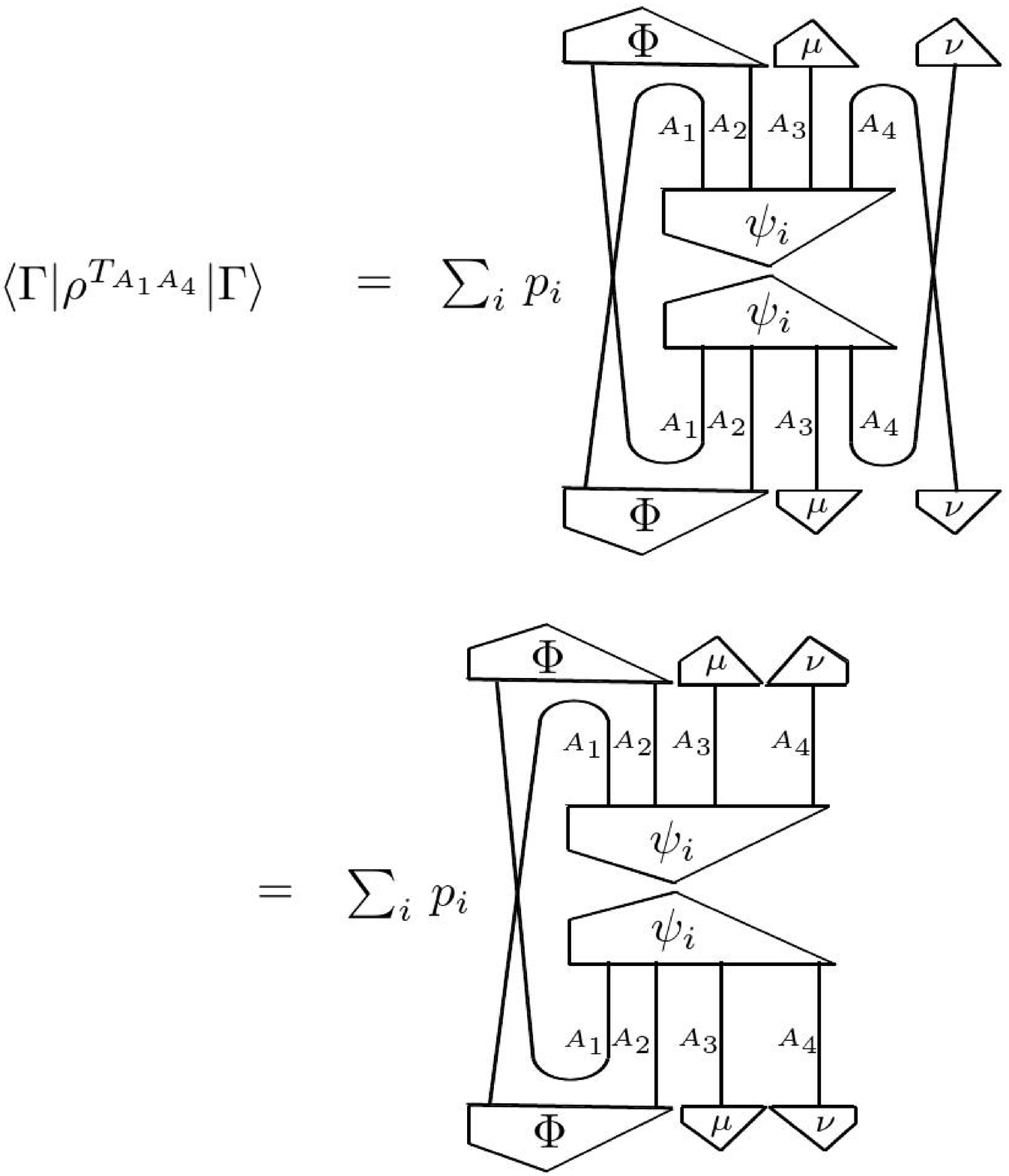}
\caption{\label{fig:ppt2} The second transpose operation, $T_{A_4}$, can be eliminated on a product state $\ket{\tau}_{A_3A_4} = \ket{\mu}_{A_3}\otimes\ket{\nu}_{A_4}$. The rest of the calculations are analogous to those on the diagram of Fig.~\ref{fig:ppt1}.}
\end{figure}

Now it can be easily seen that this case is reduced to the previous one of bipartite cut $A_1|A'_2A_4$ with the state $\ket{\tau}$ replaced with $\ket{\mu}_{A_3}\otimes\ket{\nu^*}_{A_4}$: we can repeat the reasoning starting from Eq.~(\ref{transfppt}) on and obtain that $\mathcal W_{A_1 A'_2A_4}$ is NPT across bipartite cut $A'_2|A_1A_4$. If, in addition, subspace $\mathcal{S}^{(1)}_{A_1A_2}$ is $1$-distillable, then $\mathcal W_{A_1 A'_2A_4}$ is $1$-distillable across  $A'_2|A_1A_4$.

When $n>2$, each possible bipartite cut can be analyzed similarly: choosing appropriate product  structure of the state $\ket{\Gamma}$, we reduce the case with many  transposes acting on different subsystems to the situation where there is only one partial transpose acting on some state that is entirely supported on one of the subspaces $\mathcal S$, then repeat the above reasoning. 
 \end{proof}

\nocite{*}

\bibliography{refs}

\end{document}